\let\origvec\vec
\documentclass[a4paper,envcountsect]{llncs}

\let\vec\origvec

\usepackage{amsfonts, amsmath, amssymb, alltt, color, pbox, enumerate}
\usepackage{algorithm}
\usepackage[pdftex]{graphicx}
\usepackage{tabularx}
\usepackage[title]{appendix}
\usepackage{mathtools}
\usepackage{xspace}
\usepackage{todonotes}

\title{Proper Hierarchies in Polylogarithmic Time and Absence of Complete Problems\thanks{The research reported in this paper results from the project {\em Higher-Order Logics and Structures} supported by the Austrian Science Fund (FWF: \textbf{[I2420-N31]}).}}

\author{Flavio Ferrarotti\inst{1} \and Sen\'{e}n Gonz\'{a}lez\inst{2} \and Klaus-Dieter Schewe\inst{3} \and Jos\'{e} Mar\'{\i}a Turull-Torres\inst{4}}

\institute{Software Competence Center Hagenberg, Austria
\and
P\&T Connected, Austria
\and
Zhejiang University, UIUC Institute, China
\and
Universidad Nacional de La Matanza, Buenos Aires, Argentina}

\newcommand{\dpolylog}{DPolylogTime\xspace}
\newcommand{\polylog}{\mathrm{DPolyLogTime}}
\newcommand{\npolylog}{NPolylogTime\xspace}
\newcommand{\mnpolylog}{\mathrm{NPolylogTime}}
\newcommand{\dtime}[1]{\mathrm{DTIME}(#1)}
\newcommand{\ntime}[1]{\mathrm{NTIME}(#1)}
\newcommand{\atime}[2]{\mathrm{ATIME}(#1,#2)}
\newcommand{\atimeop}[2]{\mathrm{ATIME}^{\mathit{op}}(#1,#2)}
\newcommand{\mA}{\mathbf{A}}
\newcommand{\mB}{\mathbf{B}}
\newcommand{\initial}[1]{\mathrm{InitialZeros}^{#1}}
\newcommand{\conseq}[1]{\mathrm{ConseqZeros}^{#1}}
\newcommand{\noconseq}[1]{\mathrm{NoConseqZeros}^{#1}}
\newcommand{\exactlyonce}[1]{\mathrm{ExactlyOnce}^{#1}}
\newcommand{\genconseq}[2]{\mathrm{AtLeastBlocks}^{#1}_{#2}}
\newcommand{\genexactly}[2]{\mathrm{ExactlyBlocks}^{#1}_{#2}}
\newcommand{\soplog}{\mathrm{SO}^{\mathit{plog}}}
\newcommand{\sigmaplog}[1]{\Sigma^\mathit{plog}_{#1}}
\newcommand{\piplog}[1]{\Pi^\mathit{plog}_{#1}}
\newcommand{\polysigma}[1]{\tilde{\Sigma}_{#1}^{\mathit{plog}}}
\newcommand{\polypi}[1]{\tilde{\Pi}_{#1}^{\mathit{plog}}}
\newcommand{\ttx}{{\tt x}}
\newcommand{\tty}{{\tt y}}
\newcommand{\ttz}{{\tt z}}

\begin{document}

\maketitle

\begin{abstract}
The polylogarithmic time hierarchy structures sub-linear time complexity. In recent work it was shown that all classes $\polysigma{m}$ or $\polypi{m}$ ($m \in \mathbb{N}$) in this hierarchy can be captured by semantically restricted fragments of second-order logic. In this paper the descriptive complexity theory of polylogarithmic time is taken further showing that there are strict hierarchies inside each of the classes of the hierarchy. A straightforward consequence of this result is that there are no complete problems for these complexity classes, not even under polynomial time reductions. 
As another consequence we show that the polylogarithmic time hierarchy itself is strict.
\end{abstract}

\section{Introduction}

Computations with sub-linear time complexity have not been studied intensively. However, such computations appear rather naturally, e.g. in the area of circuits. Mix Barrington studied the complexity of circuits \cite{barrington:sct1992} characterizing a class of families of constant-depth quasi-polynomial size AND/OR-circuits. In particular, he proved that the class of Boolean queries computable by the class of $\mathrm{DTIME}[(\log n)O(1)]$ DCL-uniform families of Boolean circuits of unbounded fan-in, size $2^{{(\log n)}^{O(1)}}$ and depth $O(1)$ coincides with the class of Boolean queries expressible in a fragment $SO^b$ of second-order logic. As used in his study, the complexity class $\mathrm{DTIME}[(\log n)O(1)]$ is known as {\em quasipolynomial time}. Furthermore, the fastest known algorithm for checking graph isomorphisms has also a complexity in quasipolynomial time \cite{babai:stoc2016}. 

In \cite{FerrarottiGST18} we started a deeper investigation of sub-linear time computations emphasising  complexity classes DPolyLogTime and NPolyLogTime of decision problems that can be solved deterministically or non-deterministically with a time complexity in $O(\log^k n)$ for some $k$, where $n$ is as usual the size of the input. We extended these complexity classes to a complete hierarchy, the {\em polylogarithmic time hierarchy}, analogous to the polynomial time hierarchy, and for each class $\Sigma_m^{plog}$ or $\Pi_m^{plog}$ ($m \in \mathbb{N}$) in the hierarchy we defined a fragment of semantically restricted second-order logic capturing it \cite{FerrarottiGST19,FerrarottiGTBV19}. While the hierarchy as a whole captures the same class of problems studied by Mix Barrington, the various classes of the hierarchy provide fine-grained insights into the nature of decision problems decidable in sub-linear time.

With these promising results the natural question occurs, whether there are complete problems in the hierarchy, and what would be an appropriate notion of reduction to define complete problems. Note that for the somehow related complexity class PolyLogSpace it is known since long time that it does not have complete problems. Another problem is the strictness of the polylogarithmic hierarchy.

In this paper we address this problem. 
We show that for none of the classes $\polysigma{m}$ and $\polypi{m}$ ($m \in \mathbb{N}$) in the polylogarithmic time hierarchy there exists a complete problem. It turns out that this result is a rather simple consequence of the existence of proper hierarchies inside each of the classes $\polysigma{m}$ and $\polypi{m}$. Note that a similar approach shows the non-existence of complete problems for PolyLogSpace, but the corresponding proof exploits theorems by Hartmanis et al. that cannot be applied to our case, as these theorems (which are well known in complexity theory as the space and time hierarchy theorems) require at least linear time.

The remainder of this paper is organized as follows. Section~\ref{complexityClasses} summarizes the necessary preliminaries for our investigation introducing the complexity classes of the polylogarithmic time hierarchy. This is complemented in Section~\ref{sec:plt-logic} by reviewing $\text{SO}^{plog}$, the polylogarithmically-restricted fragment of second-order logic that is used to define subsets capturing the complexity classes $\polysigma{m}$ and $\polypi{m}$. Section~\ref{sec:problems} introduces concrete decision problems that we use to show the existence of proper hierarchies inside $\polysigma{m}$ and $\polypi{m}$. We use the capturing logics to define these problems that are parametrised by $k \in \mathbb{N}$, and the various different values for $k$ give rise to the hierarchies. Theorems showing that we obtain proper hierarchies inside $\polysigma{m}$ and $\polypi{m}$ are proven in Section \ref{hierarchies}. Then the non-existence of complete problems arises as a rather straightforward consequence, as we will show in Section \ref{sec:complete}. 
We further show in Section~\ref{polyloghierarchyisstrict} that another consequence is the strictness of the polylogarithmic time hierarchy itself.
We conclude with a brief summary in Section \ref{sec:schluss}.

\section{Polylogarithmic time complexity classes}\label{complexityClasses}


The sequential access that Turing machines have to their tapes makes it impossible to compute anything in sub-linear time. Therefore, logarithmic time complexity classes are usually studied using models of computation that have random access to their input. As this also applies to the poly-logarithmic complexity classes studied in this paper, we adopt a Turing machine model that has a \emph{random access} read-only input, similar to the log-time Turing machine in~\cite{barrington:jcss1990}.

In the following, $\log n$ always refers to the binary logarithm of $n$, i.e., $\log_2 n$. With $\log^k n$ we mean $(\log n)^k$.

A \emph{random-access Turing machine} is a multi-tape Turing machine with (1) a read-only (random access) \emph{input} of length $n+1$, (2) a fixed number of read-write \emph{working tapes}, and (3) a read-write input \emph{address-tape} of length $\lceil \log n \rceil$.

Every cell of the input as well as every cell of the address-tape contains either $0$ or $1$ with the only exception of the ($n+1$)st cell of the input, which is assumed to contain the endmark $\triangleleft$. In each step the binary number in the address-tape either defines the cell of the input that is read or if this number exceeds $n$, then the ($n+1$)st cell containing $\triangleleft$ is read.   

\begin{example}\label{ex:machine}
Let polylogCNFSAT be the class of satisfiable propositional formulae in conjunctive normal form with $c \leq \lceil \log n \rceil^k$ clauses, where $n$ is the length of the formula. Note that the formulae in polylogCNFSAT tend to have few clauses and many literals. We define a random-access Turing machine $M$ which decides polylogCNFSAT. The alphabet of $M$ is $\{0,1,\#,+,-\}$. The input formula is encoded in the input tape as a list of $c \leq \lceil \log n \rceil^k$ indices, each index being a binary number of length $\lceil \log n \rceil$, followed by $c$ clauses. For every $1 \leq i \leq c$, the $i$-th index points to the first position in the $i$-th clause. Clauses start with $\#$ and are followed by a list of literals. Positive literals start with a $+$, negative with a $-$. The $+$ or $-$ symbol of a literal is followed by the ID of the variable in binary. $M$ proceeds as follows: (1) Using binary search with the aid of the ``out of range'' response $\triangleleft$, compute $n$ and $\lceil \log n \rceil$. (2) Copy the indices to a working tape, counting the number of indices (clauses) $c$. (3) Non-deterministically guess $c$ input addresses $a_1, \ldots, a_c$, i.e., guess $c$ binary numbers of length $\lceil \log n \rceil$. (4) Using $c$ $1$-bit flags, check that each $a_1, \ldots, a_c$ address falls in the range of a different clause. (5) Check that each $a_1, \ldots, a_c$ address points to an input symbol $+$ or $-$. (6) Copy the literals pointed by $a_1, \ldots, a_c$ to a working tape, checking that there are \emph{no} complementary literals. (7) Accept if all checks hold.
\end{example}

Let $L$ be a language accepted by a random-access Turing machine $M$. Assume that for some function $f$ on the natural numbers, $M$ makes at most $O(f(n))$ steps before accepting an input of length $n$. If $M$ is deterministic, then we write $L \in \dtime{f(n)}$. If $M$ is non-deterministic, then we write $L \in \ntime{f(n)}$. We define the classes of deterministic and non-deterministic poly-logarithmic time computable problems as follows:
\[ \polylog = \bigcup_{k \in \mathbb{N}} \dtime{\log^k n} \quad \, \mnpolylog = \bigcup_{k \in \mathbb{N}} \ntime{\log^k n} \]
The non-deterministic random-access Turing machine in Example~\ref{ex:machine} clearly works in polylog-time. Therefore, polylogCNFSAT $\in \mnpolylog$.

Recall that an alternating Turing machine comes with a set of states $Q$ that is partitioned into subset $Q_\exists$ and $Q_\forall$ of so-called existential and universal states. Then a configuration $c$ is accepting iff
\begin{itemize}

\item $c$ is in a final accepting state,

\item $c$ is in an existential state and there exists a next accepting configuration, or

\item $c$ is in a universal state, there exists a next configuration and all next configurations are accepting.

\end{itemize}

In analogy to our definition above we can define a \emph{random-access alternating Turing machine}. The languages accepted by such a machine $M$, which starts in an existential state and makes at most $O(f(n))$ steps before accepting an input of length $n$ with at most $m$ alternations between existential and universal states, define the complexity class $\atime{f(n)}{m}$. Analogously, we define the complexity class $\atimeop{f(n)}{m}$ comprising languages that are accepted by a random-access alternating Turing machine that starts in a universal state and makes at most $O(f(n))$ steps before accepting an input of length $n$ with at most $m-1$ alternations between universal and existential states. With this we define
\[ \polysigma{m} = \bigcup_{k \in \mathbb{N}} \mathrm{ATIME}[\log^k n,m] \quad \text{and} \quad \polypi{m} = \bigcup_{k \in \mathbb{N}} \mathrm{ATIME}^{op}[\log^k n,m] . \]

The poly-logarithmic time hierarchy is then defined as $\mathrm{PLH} = \bigcup_{m \ge 1} \polysigma{m}$. Note that $\polysigma{1} = \mnpolylog$ holds. 

\begin{remark}

Note that a simulation of a $\mnpolylog$ Turing machine $M$ by a deterministic machine $N$ requires checking all computations in the tree of computations of $M$. As $M$ works in time $({\log n})^{O(1)}$, $N$ requires time $2^{{\log n}^{O(1)}}$. This implies $\mnpolylog \subseteq \mathrm{DTIME}(2^{{\log n}^{O(1)}})$, which is the complexity class called quasipolynomial time of the fastest known algorithm for graph isomorphism \cite{babai:stoc2016}, which further equals the class  
$\mathrm{DTIME}({n^{{\log n}^{O(1)}}})$\footnote{This relationship appears quite natural in view of the well known relationship $\mathrm{NP} = \mathrm{NTIME}(n^{O(1)}) \subseteq \mathrm{DTIME}(2^{{n}^{O(1)}}) = \mathrm{EXPTIME}$.}.
\end{remark}

\section{Logics for polylogarithmic time}\label{sec:plt-logic}

The descriptive complexity of the polylogarithmic time complexity classes described in the previous section, has been recently studied in deepth in~\cite{FerrarottiGST18,FGST18,FerrarottiGST19,FerrarottiGTBV19,FerrarottiGTBV19b}, where precise logical characterization of those classes were presented. The logics used in those characterizations are quite useful to think and describe the problems used in this paper to prove proper hierarchies inside polylogarithmic time. In this section we describe these logics and the results regarding their correspondence with the different polylogarithmic time complexity classes.

The capturing results for polylogarithmic time hold over ordered structures. A finite ordered $\sigma$-structure $\mA$ is a finite structure of vocabulary $\sigma\cup\{<\}$, where $\leq\notin\sigma$ is a binary relation symbol and $<^\mA$ is a linear order on $A$.
Every finite ordered structure has a corresponding isomorphic structure, whose domain is an
initial segment of the natural numbers. Thus, we assume, as usual, that $A = \{0, 1, \ldots, n-1\}$, where $n$ is the cardinality $|A|$ of $A$.  In the case of non-deterministic polylogarithmic time complexity, the capturing results also assume that $\sigma$ includes $\mathrm{SUCC}$, $\mathrm{BIT}$ and constants for $\log n$, the minimum, second and maximum elements. In every structure $\bf A$, the symbol $\mathrm{SUCC}$ is interpreted by the successor relation corresponding to the $<^{\bf A}$ ordering. The constant symbols $0$, $1$ and $\mathit{max}$ are in turn interpreted as the minimum, second and maximum elements under the $<^{\bf A}$ ordering and the constant $\mathit{logn}$ as $\left\lceil \log |A| \right\rceil$. Finally, $\mathrm{BIT}$ is interpreted by the following binary relation:
\[\mathrm{BIT}^{\bf A} = \{(i, j) \in A^2 \mid \text{Bit $j$ in the binary representation of $i$ is $1$}\}.\]
W.l.o.g., we assume that all structures have at least \emph{three} elements. This results in a cleaner presentation, avoiding trivial cases which would unnecessarily complicate some formulae.

Let us start with \dpolylog. This class is captured by the \emph{index logic} introduced in~\cite{FerrarottiGTBV19}. Index logic is two-sorted; variables of the first sort range over
the domain of the input structure.  Variables of the second sort
range over an initial segment of the natural numbers; this
segment is bounded by the logarithm of the size of the input
structure.  Thus, the elements of the second sort represent the
bit positions needed to address elements of the first sort.
Index logic includes full fixpoint logic on the second sort.
Quantification over the first sort, however, is heavily
restricted.  Specifically, a variable of the first sort can only
be bound using an address specified by a subformula that defines
the positions of the bits of the address that are set.
This ``indexing mechanism'' lends index logic its name.

The following result confirms that the problems that can be described in the index logic are in $\dpolylog$ and vice versa.
    \begin{theorem}[\cite{FerrarottiGTBV19}]\label{captureResult}    
    Index logic captures $\dpolylog$ over ordered structures.
    \end{theorem}

Regarding nondeterministic polylogarithmic time, the restricted second-order logic $\soplog$ defined in~\cite{FerrarottiGST18,FGST18,FerrarottiGST19} captures the polylogarithmic-time hierarchy, with its quantifier prenex fragments $\sigmaplog{m}$ and $\piplog{m}$ capturing  the corresponding levels $\polysigma{m}$ and $\polypi{m}$ of this hierarchy, respectively.

$\soplog$ is a fragment of second-order logic where second-order quantification range over relations of polylogarithmic size and first-order quantification is restricted to the existential fragment of first-order logic plus universal quantification over variables under the scope of a second-order variable. 

Formally, we can inductively define the syntax of $\soplog$ as follows:
\begin{itemize}
\item Every formula in the existential fragment of first-order logic with equality is a $\soplog$ formula.

\item If $X$ is a second-order variable of arity $r$, and $t_1, \ldots, t_r$ are first-order terms, then both $X(t_1, \ldots, t_r)$ and $X(t_1, \ldots, t_r)$ are $\soplog$ formulae.

\item If $\varphi$ and $\psi$ are are $\soplog$ formulae, then $(\varphi \wedge \psi)$ and $(\varphi \vee \psi)$ are are $\soplog$ formulae.
    
\item If $\varphi$ is a $\soplog$ formula, $X$ is a second-order variable of arity $r$ and $\bar{x}$ is an $r$-tuple of first-order variables, then $\forall \bar{x} (X(\bar{x}) \rightarrow \varphi)$ is $\soplog$ formula.    
    
\item If $\varphi$ is a $\soplog$ formula and $x$ is a first-order variable, then $\exists x \varphi$ is a $\soplog$ formula.   
     
\item If $\varphi$ is a $\soplog$ formula and $X$ is a second-order variable, then both $\exists X \varphi$ and $\forall X \varphi$ are $\soplog$ formulae.
\end{itemize}    

The most significant restriction of $\soplog$ is in its semantics. In addition to its arity, each second-order variable $X$ is associated with another non-negative integer, its \emph{exponent}, and it is required that any $X$ of arity $r$ and exponent $k$ is interpreted on a structure of domain $A$ as an $r$-ary relation \emph{of cardinality smaller than} $\log^k |A|$. Otherwise, the semantics of $\soplog$ follows the standard semantics of second-order logic. 

As usual, the fragments $\sigmaplog{m}$ (resp. $\piplog{m}$) are defined by considering $\soplog$ formulae with $m$ alternating blocks of second-order quantifiers in quantifier prenex (Skolem) normal form, starting with an existential (resp. universal) block. Note that by Lemma~3 in~\cite{FGST18}, for every $\soplog$ formula $\varphi$ there is an equivalent formula $\varphi'$ that is in quantifier prenex normal form. In the following we will assume that the reader is familiar with the techniques that can be applied to transform arbitrary $\soplog$ formulae into equivalent formulae in Skolem normal form. Those techniques are detailed in the proof of Lemma~3 in Appendix~B in~\cite{FGST18}. 

The following result characterizes precisely the expressive power of $\soplog$ in terms of the nondeterministic polylogarithmic time hierarchy. Note that in particular, existential $\soplog$ captures $\npolylog$. 

\begin{theorem}[\cite{FerrarottiGST18,FerrarottiGST19}]\label{pedsoplog3}
Over ordered structures with successor relation, $\mathrm{BIT}$ and constants for $\log n$, the minimum, second and maximum elements, $\Sigma^{\mathit{plog}}_m$ captures $\tilde{\Sigma}^{\mathit{plog}}_m$ and $\Pi^{\mathit{plog}}_m$ captures $\tilde{\Pi}^{\mathit{plog}}_m$ for all $m \ge 1$.
\end{theorem}

\section{Problems that lead to proper hierarchies}\label{sec:problems}

Here we introduce the decision problems that we use in the next section to show the existence of proper hierarchies of polylogarithmic-time. In addition, for the nonseterministic classes we give a precise definition of these problems in terms of the logic $\soplog$ studied in~\cite{FerrarottiGST18,FGST18,FerrarottiGST19} and discussed in the previous section.  

From now on we work with the class of structures known as \emph{word models} (see for instance~\cite{EbbinghausF95}).
Let $\pi$ be the vocabulary $\{<, R_0, R_1\}$, where $<$ is a binary relation symbol and $R_0, R_1$ are unary relation symbols. We can identify any binary string (word) $w = a_1 \ldots a_{n}$ in $\{0,1\}^+$ with a $\pi$-structure (word model) ${\mA}_w$, where the cardinality of the domain $A$ of ${\mA_w}$ equals the length of $w$, $<^{{\mA}_w}$ is a linear order in $A$, $R_0^{{\mA}_w}$ contains the positions in $w$ carrying a $0$, and $R_1^{{\mA}_w}$ contains the positions in $w$ carrying a $1$.

\begin{problem}[$\initial{k}$]\label{p0}
The problem $\initial{k}$ consists on deciding (over word models of signature $\pi$) the language of binary strings which have a prefix of at least $\lceil\log^{k} n\rceil$ consecutive zeros, where $n$ is the length of the string.

\end{problem}

\begin{problem}[$\conseq{k}$]\label{p1}
Let $\conseq{k}$ denote the problem of deciding the language of binary strings which have at least $\lceil\log^{k} n\rceil$ consecutive bits set to $0$, where $n$ is the length of the string. This can be expressed formally in $\soplog$ as follows:\\[0.2cm]
$\exists X (|X| = \log^kn \wedge \mathit{SEQ}(X) \wedge \forall x (X(x) \rightarrow R_0(x))$,\\[0.2cm]
where $X$ is of arity $1$ and exponent $k$, the expression $|X| = \log^kn$ denotes the sub-formula which defines that the cardinality of $X$ is $\lceil \log n \rceil^k$, and $\mathit{SEQ}(X)$ denotes the sub-formula expressing that the elements of $X$ are a contiguous subsequence of the order $<$. 

The sub-formula expressing $|X| = \log^kn$ can be written as follows:\\[0.2cm]
$\exists Y \bar{x} (Y(\bar{x}) \wedge \bar{x} = \bar{0} \wedge \forall \bar{y} (Y(\bar{y}) \rightarrow (\mathit{SUCCk}(\bar{y}, \overline{\mathit{logn}}) \vee \exists \bar{z} (Y(\bar{z}) \wedge \mathit{SUCCk}(\bar{y}, \bar{z}))) \wedge $\\
\hspace*{1cm}$|X| = |Y|)$\\[0.2cm]
where $Y$ is of arity $k$ and exponent $k$, $\bar{x}, \bar{y}, \bar{z}$ denote $k$-tuples of first-order variables, $\mathit{SUCCk}(\bar{y}, \bar{z})$ denotes a sub-formula expressing that $\bar{z}$ is the immediate successor of $\bar{y}$ in the lexicographical order of $k$-tuples, and $|X| = |Y|$ expresses that $X$ and $Y$ have equal cardinality. $\mathit{SUCCk}(\bar{y}, \bar{z})$ can be expressed by a quantifier-free $\soplog$ formula (for details refer to $\mathit{SUCC}_k$ in Section~4 in~\cite{FGST18}). In turn, $|X| = |Y|$ can be expressed by an existential $\soplog$ formula using second order variables of arity $k+1$ and exponent $k$ (for details refer to Section~3.1 in~\cite{FGST18}).

Finally, $\mathit{SEQ}(X)$ can be expressed in $\soplog$ as follows:\\[0.2cm]
$\forall x (X(x) \rightarrow \exists y (\mathit{SUCC}(x, y) \vee \forall z (X(z) \rightarrow z < x)))$\\[0.2cm]
The whole formula for $\conseq{k}$ can then be rewritten in Skolem normal form as a formula in $\sigmaplog{1}$ with second order variables of exponent $k$. 
\end{problem}

\begin{problem}[$\noconseq{k}$]\label{p2}
Let $\noconseq{k}$ denote the problem of deciding the language of binary strings which do \emph{not} have greater than or equal $\lceil\log^{k} n\rceil$ consecutive bits set to $0$, where $n$ is the length of the string. Since syntactically the negation of a formula in $\soplog$ is not always a formula in $\soplog$, we cannot just negate the formula for $\conseq{k+1}$ in Problem~\ref{p1} to get the $\soplog$ formula for $\noconseq{k}$. We can nevertheless define $\noconseq{k}$ as follows: \\[0.2cm]
$\forall X (|X| = \log^kn \wedge \mathit{SEQ}(X)  \rightarrow \exists x (X(x) \wedge R_1(x)))$\\[0.2cm]
This is equivalent to:\\[0.2cm] 
$\forall X (\neg(|X| = \log^kn) \vee \neg\mathit{SEQ}(X)  \vee \exists x (X(x) \wedge R_1(x))$.\\[0.2cm]
It follows that the negations of the sub-formulae $|X| = \log^kn$ that we defined in Problem~\ref{p1} is in $\piplog{1}$. Regarding $\neg\mathit{SEQ}(X)$, it can be written in $\soplog$ as\\[0.2cm]
$\exists x y z (X(x) \wedge \neg X(y) \wedge X(z) \wedge x < z \wedge \mathrm{SUCC}(x,y))$.\\[0.2cm] 
We then get that the formula for $\noconseq{k}$ can be rewritten in Skolem normal form as a formula in $\piplog{1}$ with second order variables of exponent $k$.   
\end{problem}

\begin{problem}[$\exactlyonce{k}$]\label{p3}
Let $\exactlyonce{k}$ denote the problem of deciding the language of binary strings which contain the substring $0^{\lceil \log n \rceil^k}$ exactly once, i.e., $s$ is in $\exactlyonce{k}$ iff $0^{\lceil \log n \rceil^k}$ is a substring of $s$ and every other substring of $s$ is not $0^{\lceil \log n \rceil^k}$.
This can be expressed formally in $\soplog$ by combining the formulae for $\conseq{k}$ and $\noconseq{k}$ (see Problems~\ref{p1} and~\ref{p2}, respectively) as follows:\\[0.2cm]
$\exists X (|X| = \log^kn \wedge \mathit{SEQ}(X) \wedge \forall x (X(x) \rightarrow R_0(x)) \wedge$\\[0.2cm]
\hspace*{1cm} $\forall Y (Y = X \vee \neg(|Y| = \log^kn) \vee \neg\mathit{SEQ}(Y) \vee \exists x (X(x) \wedge R_1(x))))$,\\[0.2cm]
Clearly, all second order variables in the formula need maximum exponent $k$ and the formula itself can be rewritten in Skolem normal form as a formula in $\sigmaplog{2}$.
\end{problem}

\begin{problem}[$\genconseq{k}{l}$]\label{p5}
Let $\genconseq{k}{l}$ for $m \geq 0$ denote the problem of deciding the language of binary strings with at least $(\lceil\log n\rceil^k)^l$ non-overlapping substrings of the form $0^{\lceil \log n \rceil^k}$, where $n$ is the length of the string. 

If $l = 0$ then this is equivalent to $\conseq{k}$ and, as discussed in Problem~\ref{p1}, it can be expressed in $\sigmaplog{1}$.  

If $l = 1$, we can express $\genconseq{k}{l}$ in $\soplog$ as follows:\\[0.2cm]
$\exists X \forall x y \exists Z (|X| = \log^kn \wedge \mathit{SEQP}(X) \wedge $\\[0.1cm]
\hspace*{1cm}$(X(x,y) \rightarrow (|Z| = \log^kn \wedge \mathit{SEQ}(Z) \wedge \mathit{min}(Z) = x \wedge \mathit{max}(Z) = y  \wedge$\\[0.1cm]
\hspace*{3cm} $\forall z (Z(z) \rightarrow R_0(z)))))$.\\[0.2cm]
Here $\mathit{SEQP}(X)$ denotes the sub-formula expressing that $X$ is a set of ordered pairs that form a sequence where every consecutive $(a_1, a_2)$ and $(b_1, b_2)$ in the sequence satisfy that $a_2$ is the immediate predecessor of $b_1$ in the order $<$. This is clearly expressible by a $\soplog$ formula free of second-order quantification. The sub-formulae $\mathit{min}(Z) = x$ and  $\mathit{max}(Z) = y$ have the obvious meaning and again can easily be expressed in $\soplog$ without using second-order quantification. The whole sentence can be transformed into an equivalent sentence in $\sigmaplog{3}$.

Finally, for every $l \geq 2$, we can express $\genconseq{k}{l}$ in $\soplog$ with formulae of the form:\\[0.2cm]
$\exists X_1 \forall x_1 y_1 \exists X_2 \forall x_2 y_2 \cdots \exists X_l \forall x_l y_l \exists Z (|X_1| = \log^kn \wedge \mathit{SEQP}(X_1) \wedge$\\[0.1cm]
$(X_1(x_1,y_1) \rightarrow$\\[0.1cm]
\hspace*{0.2cm} $(|X_2| = \log^kn \wedge \mathit{SEQP}(X_2) \wedge \mathit{minp}(X_2) = x_1 \wedge \mathit{maxp}(X_2) = y_1  \wedge$\\[0.1cm]
\hspace*{0.4cm} $\cdots \wedge (X_{l-1}(x_{l-1},y_{l-1}) \rightarrow$\\[0.1cm]
\hspace*{1.5cm} $(|X_l| = \log^kn \wedge \mathit{SEQP}(X_l) \wedge \mathit{minp}(X_l) = x_{l-1} \wedge \mathit{maxp}(X_l) = y_{l-1}  \wedge$\\[0.1cm]
\hspace*{2cm} $(X_l(x_l,y_l) \rightarrow$\\[0.1cm] 
\hspace*{2.5cm}$(|Z| = \log^kn \wedge \mathit{SEQ}(Z) \wedge \mathit{min}(Z) = x_l \wedge \mathit{max}(Z) = y_1  \wedge $\\[0.1cm]
\hspace*{2.8cm} $\forall z (Z(z) \rightarrow R_0(z))))))\cdots)))$.\\[0.2cm]
The sub-formulae of the form $\mathit{minp}(X) = x$ (resp. $\mathit{maxp}(X) = x$) express that $x$ is the smallest first element (resp. biggest second element) of any tuple in $X$ and is easily expressible in $\soplog$ by a formula free of second-order quantifiers. We can rewrite the whole formula as a $\sigmaplog{2 \cdot l + 1}$ formula.  
\end{problem}

\begin{problem}[$\genexactly{k}{l}$]\label{p6}
Let $\genexactly{k}{l}$ for $m \geq 0$ denote the problem of deciding the language of binary strings with exactly $(\lceil\log n\rceil^k)^l$ non-overlapping substrings of the form $0^{\lceil \log n \rceil^k}$, where $n$ is the length of the string. 

If $l = 0$ then this is equivalent to $\exactlyonce{k}$ and, as discussed in Problem~\ref{p3}, it can be expressed in $\sigmaplog{2}$.  

If $l = 1$, we can express $\genexactly{k}{l}$ in $\soplog$ as follows:\\[0.2cm]
$\exists X \forall x y \exists Z (|X| = \log^kn \wedge \mathit{SEQP}(X) \wedge $\\[0.1cm]
\hspace*{1cm}$(X(x,y) \rightarrow (|Z| = \log^kn \wedge \mathit{SEQ}(Z) \wedge \mathit{min}(Z) = x \wedge \mathit{max}(Z) = y  \wedge$\\[0.1cm]
\hspace*{3cm} $\forall z (Z(z) \rightarrow R_0(z)) \wedge$\\[0.1cm]
\hspace*{3cm} $\forall X' \exists x'y' \forall Z'(X' = X \vee \neg(|X'| = \log^kn) \vee \neg \mathit{SEQP}(X') \vee$\\[0.1cm]
\hspace*{3.5cm} $(X'(x',y') \wedge (Z' = Z \vee \neg(|Z'| = \log^kn) \vee \neg\mathit{SEQ}(Z') \vee$\\[0.1cm]
\hspace*{3.5cm} $\neg(\mathit{min}(Z') = x') \vee \neg(\mathit{max}(Z') = y') \vee $\\[0.1cm]
\hspace*{3.5cm} $\exists z' (Z'(z') \wedge R_1(z')))))))$.\\[0.2cm]
It is not difficult to see that this formula can be rewritten as a $\sigmaplog{4}$ formula. 

Finally, for every $l \geq 2$, we can express $\genexactly{k}{l}$ in $\soplog$ with formulae of the form:\\[0.2cm] 
$\exists X_1 \forall x_1 y_1 \exists X_2 \forall x_2 y_2 \cdots \exists X_l \forall x_l y_l \exists Z (|X_1| = \log^kn \wedge \mathit{SEQP}(X_1) \wedge$\\[0.1cm]
$(X_1(x_1,y_1) \rightarrow$\\[0.1cm]
\hspace*{0.2cm} $(|X_2| = \log^kn \wedge \mathit{SEQP}(X_2) \wedge \mathit{minp}(X_2) = x_1 \wedge \mathit{maxp}(X_2) = y_1  \wedge$\\[0.1cm]
\hspace*{0.4cm} $\cdots \wedge (X_{l-1}(x_{l-1},y_{l-1}) \rightarrow$\\[0.1cm]
\hspace*{1.5cm} $(|X_l| = \log^kn \wedge \mathit{SEQP}(X_l) \wedge \mathit{minp}(X_l) = x_{l-1} \wedge \mathit{maxp}(X_l) = y_{l-1}  \wedge$\\[0.1cm]
\hspace*{2cm} $(X_l(x_l,y_l) \rightarrow$\\[0.1cm] 
\hspace*{2.5cm}$(|Z| = \log^kn \wedge \mathit{SEQ}(Z) \wedge \mathit{min}(Z) = x_l \wedge \mathit{max}(Z) = y_l  \wedge $\\[0.1cm]
\hspace*{2.8cm} $\forall z (Z(z) \rightarrow R_0(z)) \wedge$\\[0.1cm]
\hspace*{2.8cm} $\forall X_1' \exists x_1' y_1' \forall X_2' \exists x_2' y_2' \cdots \forall X_l' \exists x_l' y_l' \forall Z'(X_1' = X_1 \vee$\\[0.1cm]
\hspace*{3cm} $\neg(|X_1'| = \log^kn) \vee \neg\mathit{SEQP}(X_1') \vee (X_1'(x_1',y_1') \wedge (X_2' = X_2 \vee $\\[0.1cm]
\hspace*{3cm} $\neg(|X_2'| = \log^kn) \vee \neg\mathit{SEQP}(X_2') \vee \neg\mathit{minp}(X_2') = x_1' \vee$\\[0.1cm]
\hspace*{3cm} $\neg\mathit{maxp}(X_2') = y_1'  \vee ( \cdots \vee (X_{l-1}'(x_{l-1}',y_{l-1}') \wedge (X_{l}' = X_{l} \vee$\\[0.1cm]
\hspace*{3cm} $\neg(|X_l'| = \log^kn) \vee \neg\mathit{SEQP}(X_l') \vee \neg(\mathit{minp}(X_l') = x_{l-1}') \vee$\\[0.1cm] \hspace*{3cm} $\neg(\mathit{maxp}(X_l') = y_{l-1}')  \vee (X_l'(x_l',y_l') \wedge (Z' = Z \vee $\\[0.1cm]
\hspace*{3cm} $\neg(|Z'| = \log^kn) \vee \neg\mathit{SEQ}(Z') \vee \neg(\mathit{min}(Z') = x_l') \vee$\\[0.1cm]
\hspace*{3cm} $\neg(\mathit{max}(Z') = y_l') \vee \exists z' (Z'(z') \wedge R_1(z')))))) \cdots ))))))))\cdots)))$.\\[0.2cm]
We can rewrite formulae of this form as $\sigmaplog{2 \cdot l + 2}$ formulae.  
\end{problem}

\section{Proper hierarchies in polylogarithmic time}\label{hierarchies}

We now present the key results of the paper showing that all the polylogarithmic complexity classes defined in Section~\ref{complexityClasses}, including every level of the polylogarithmic time hierarchy, contain proper hierarchies defined in terms of the smallest degree of the polynomial required for the decision problems introduced in the previous section.

In order to relate the problems described in the previous section using logics to the polylogarithmic complexity classes defined in terms of random-access Turing machines, we adhere to the usual conventions concerning binary encoding of finite structures~\cite{Immerman99}. That is, if $\sigma = \{R^{r_1}_1, \ldots, R^{r_p}_p, c_1, \ldots, c_q\}$ is a vocabulary, and ${\bf A}$ with $A = \{0, 1, \ldots, n-1\}$ is an ordered structure of vocabulary $\sigma$. Each relation $R_i^{\bf A} \subseteq A^{r_i}$ of $\bf A$ is encoded as a binary string $\mathrm{bin}(R^{\bf A}_i)$ of length $n^{r_i}$ where $1$ in a given position indicates that the corresponding tuple is in $R_i^{\textbf{A}}$.
Likewise, each constant number $c^{\bf A}_j$ is encoded as a binary string $\mathrm{bin}(c^{\bf A}_j)$ of length $\lceil \log n \rceil$. The encoding of the whole structure $\mathrm{bin}(\textbf{A})$ is simply the concatenation of the binary strings encodings its relations and constants. The length $\hat{n} = |\mathrm{bin}(\textbf{A})|$ of this string is $n^{r_1}+\cdots+n^{r_p} + q \lceil \log n \rceil$, where $n = |A|$ denotes the size of the input structure ${\bf A}$. Note that $\log \hat{n} \in O(\lceil \log n \rceil)$, so $\mathrm{NTIME}[\log^k \hat{n}] = \mathrm{NTIME}[\log^k n]$ (analogously for $\mathrm{DTIME}$). Therefore, we can consider random-access Turing machines, where the input is the encoding $\mathrm{bin}(\textbf{A})$ of the structure \textbf{A} followed by the endmark $\triangleleft$.

The following simple lemmas are useful to prove our hierarchy theorems. They show that the problems in the previous section can be expressed by random-access machines working in the required levels of the hierarchy theorems. 

\begin{lemma}\label{lemma:dpolylog}
$\initial{k}$ (see Problem~\ref{p0}) can be decided in $\dtime{\log^{k}n}$. 
\end{lemma}
 
\begin{proof}
Assume the input tape encodes a word model $\mA$ of signature $\pi$, i.e., a binary string. A deterministic random-access Turing machine can in deterministic time $O(\log n)$ calculate and write in its index-tape the address of the first bit in the encoding of $R_0^{\mA}$. Then it only needs to check whether this bit and the subsequent $\lceil \log n \rceil^k -1$ bits in the input-tape are $1$. If that is the case, then the machine accepts the input. Clearly, this process takes time $O(\log^k n)$.\qed
\end{proof}

\begin{lemma}\label{lemma:npolylog}
$\conseq{k}$ (see Problem~\ref{p1}) can be decided in $\ntime{\log^{k}n}$. 
\end{lemma}

\begin{proof}
Assume the input tape encodes a word model $\mA$ of signature $\pi$. A random-access Turing machine $M$ can non-deterministically guess a position $i$ in the input tape which falls within the cells encoding $R^{\mA}_0$. This takes time $O(\log n)$. Then $M$ can check (working deterministically) in time $O(log^{k+1} n)$ whether each cell of the input tape between positions $i$ and $i+log^{k+1} n$ has a $0$.\qed
\end{proof}

\begin{lemma}\label{lemma:pi1}
$\noconseq{k}$ (see Problem~\ref{p2}) can be decided in \\ $\atimeop{\log^k n}{1}$.
\end{lemma}

\begin{proof}
Assume the input tape encodes a word model $\mA$ of signature $\pi$. In a universal state, a random-access alternating Turing machine $M$ can check whether for all cell in some position $i$ in the input tape which falls in a position encoding $R^{\mA}_0$ and is at distance at least $\lceil \log n \rceil^k$ from the end of the encoding, there is a position between positions $i$ and $i+log^{k} n$ with $0$. Each of these checking can be done deterministically in time $O(log^{k} n)$. Therefore this machine decides $\noconseq{k}$ in $\atimeop{\log^{k} n}{1}$.\qed   
\end{proof}

\begin{lemma}\label{lemma:sigma2}
$\exactlyonce{k}$ (see Problem~\ref{p3}) can be decided in $\atime{\log^k n}{2}$.
\end{lemma}

\begin{proof}
We only need to combine the machines that decide $\conseq{k}$ and $\noconseq{k}$ in Lemmas~\ref{lemma:npolylog} and~\ref{lemma:pi1}, respectively. Thus, an alternating random-access Turing machine machine $M$ can decide $\exactlyonce{k}$ as follows: Assume the input tape encodes a word model $\mA$ of signature $\pi$. Let $s$ and $t$ be the cells that mark the beginning and end of the encoding of $R^{\mA}$. These cells can be calculated by $M$ in $DTIME(\log n)$. First $M$ checks in an existential state whether there is a position $i$ in the input tape which fall between $s$ and $t-log^{k}$ such that each cell between positions $i$ and $i+log^{k} n$ has a $1$. Then $M$ switches to a universal state and checks whether for all cell in some position $j$ between $s$ and $t - \log^{k} n$ of the input tape other than position $i$, there is a cell between positions $j$ and $j+log^{k} n$ with $0$. If these two checks are successful, then the input string belongs $\exactlyonce{k}$. We already saw in Lemmas~\ref{lemma:npolylog} and~\ref{lemma:pi1} that both checks can be done in time $O(\log^{k} n)$.\qed
\end{proof}

In order to get tighter upper bounds, in the previous lemmas we explicitly defined the random-access Turing machines that decide the problems. For the following two lemmas we use the upper bounds resulting from the proof of Theorem~\ref{pedsoplog3} instead, since there seems to be no better upper bounds for these cases. Thus, Lemma~\ref{lemma:genconseq} and~\ref{lemma:genexactly} follow from the facts that: (a) to evaluate the $\soplog$ formulae in Problems~\ref{p5} and~\ref{p6} for $\genconseq{k}{l}$ and $\genexactly{k}{l}$, respectively, the machine needs (as shown in the proof of Theorem~\ref{pedsoplog3} in~\cite{FGST18}) to ``guess'' $\lceil \log n \rceil^k$  addresses, each of length $\lceil \log n \rceil$; and (b) the formula for $\genconseq{k}{l}$ and $\genexactly{k}{l}$ are in $\sigmaplog{2 \times l + 1}$ and $\sigmaplog{2 \times l + 2}$, respectively. 

\begin{lemma}\label{lemma:genconseq}
For $l \geq 0$, $\genconseq{k}{l}$ (see Problem~\ref{p5}) can be decided in\\ $\atime{\log^{k+1} n}{2 \cdot l +1}$.
\end{lemma}

\begin{lemma}\label{lemma:genexactly}
For $l \geq 0$, $\genexactly{k}{l}$ (see Problem~\ref{p6}) can be decided in\\ $\atime{\log^{k+1} n}{2 \cdot l + 2}$.
\end{lemma}


We can now prove our first hierarchy theorem which shows that there is a strict hierarchy of problems inside \dpolylog.

\begin{theorem}\label{strictDet}
For every $k > 1$, $\dtime{\log^k n} \subsetneq \dtime{\log^{k+1} n}$.
\end{theorem}
\begin{proof}
Lemma~\ref{lemma:dpolylog} proves that $\initial{k+1} \in \dtime{\log^{k+1}n}$. Regarding the lower bound, we will show that $\initial{k+1}$ (see Problem~\ref{p0}) is \emph{not} in $\dtime{\log^{k}n}$. 

Let us assume for the sake of contradiction that there is a deterministic random-access Turing machine $M$ that decides $\initial{k+1}$ in time $\lceil \log n \rceil^k \cdot c$, for some constant $c \geq 1$. Take a string $s$ of the form $0^n$ such that $\lceil \log n \rceil^{k+1} >  \lceil \log n \rceil^k \cdot c$. Let $\mA$ be its corresponding word model. Since the running time of $M$ on input $\mA$ is strictly less than $\lceil\log n\rceil^{k+1}$, then there must be at least one position $i$ among the first $\lceil\log n\rceil^{k+1}$ cells in the encoding of $R_0^{\mA}$ in the input tape that was not read in the computation of $M(\mA)$. Define a string $s' = 0^i10^{n-i-1}$ and a corresponding word model $\mB$. Clearly, the output of the computations of $M(\mA)$ and $M(\mB)$ are identical. This contradicts the assumption that $M$ decides $\initial{k+1}$, since it is not true that the first $\lceil\log n\rceil^{k+1}$ bits of $s'$ are $0$.    
\qed
\end{proof}

Our second hierarchy theorem shows that there is also a strict hierarchy of problem inside \npolylog.

\begin{theorem}\label{strictNonD}
    For every $k > 1$,  $\ntime{\log^k n} \subsetneq \ntime{\log^{k+1} n}$.
\end{theorem}

\begin{proof}
Lemma~\ref{lemma:npolylog} proves that $\conseq{k+1} \in \ntime{\log^{k+1}n}$. Regarding the lower bound, we will show that $\conseq{k+1}$ (see Problem~\ref{p1}) is \emph{not} in $\ntime{\log^{k}n}$.

Let us assume for the sake of contradiction that there is a nondeterministic random-access Turing machine $M$ that decides $\conseq{k+1}$ in time $\lceil \log n \rceil^k \cdot c$, for some constant $c \geq 1$. Take a binary string $s$ of the form $0^{\lceil \log n \rceil^{k+1}}1^{n-\lceil \log n \rceil^{k+1}}$ such that $\lceil \log n \rceil^{k+1} >  \lceil \log n \rceil^k \cdot c$. Let $\mA$ be its corresponding word model. Since $M$ accepts $\mA$, then there is at least one computation $\rho$ of $M$ which accepts $\mA$ in at most $\lceil \log n \rceil^k \cdot c$ steps. Then there must be at least one position $i$ among the first $\lceil \log n \rceil^{k+1}$ cells in the encoding of $R_0^{\mA}$ in the input tape that was not read during computation $\rho$. Define a string $s' = 0^{i}10^{\lceil \log n \rceil^{k+1}-i-1}1^{n-\lceil \log n \rceil^{k+1}}$ and a corresponding word model $\mB$. Clearly, the accepting computation $\rho$ of $M(\mA)$ is also an accepting computation of $M(\mB)$. This contradicts the assumption that $M$ decides $\conseq{k+1}$, since it is not true that there are $\lceil \log n \rceil^{k+1}$ consecutive zeros in $s'$. 
\qed
\end{proof}

The following theorem shows that there is a strict hierarchy of problems inside the fist level of the $\tilde{\Pi}_m^{\mathit{plog}}$ hierarchy. 

\begin{theorem}\label{hierachyInPi0}
    For every $k > 1$,  $\atimeop{\log^k n}{0} \subsetneq \atimeop{\log^{k+1} n}{1}$.
\end{theorem}

\begin{proof}
Lemma~\ref{lemma:pi1} proves that $\noconseq{k+1} \in \atimeop{\log^{k+1} n}{1}$. Regarding the lower bound, we will show that $\noconseq{k+1}$ (see Problem~\ref{p2}) is \emph{not} in $\atimeop{\log^k n}{1}$.

Let us assume for the sake of contradiction that there is an alternating random-access Turing machine $M$ that decides $\noconseq{k+1}$ using only universal states and in time $\lceil \log n \rceil^k \cdot c$, for some constant $c \geq 1$. Take a binary string $s$ of the form $0^{\lceil \log n \rceil^{k+1}}1^{n-\lceil \log n \rceil^{k+1}}$ such that $\lceil \log n \rceil^{k+1} > \lceil \log n \rceil^k \cdot c$. Let $\mA$ be its corresponding word model. 
From our assumption that $M$ decides $\noconseq{k+1}$, we get that there is a rejecting computation $\rho$ of $M(\mA)$. Since every computation of $M$ which rejects $\mA$ must do so reading at most $\lceil \log n \rceil^k \cdot c$ cells, then there must be at least one position $i$ among the first $\lceil \log n \rceil^{k+1}$ cells in the encoding of $R_0^{\mA}$ in the input tape that was not read during computation $\rho$.
Define a string $s' = 0^i10^{{\lceil \log n \rceil^{k+1}} - i - 1}1^{n-\lceil \log n \rceil^{k+1}}$ and a corresponding word model $\mB$. Clearly, the rejecting computation $\rho$ of $M(\mA)$ is also a rejecting computation of $M(\mB)$. This contradicts the assumption that $M$ decides $\noconseq{k+1}$, since $s'$ do not have $\lceil\log n\rceil^{k+1}$ consecutive bits set to $0$ and should then be accepted by all computations of $M$.   
\qed
\end{proof}

The following theorem shows that there is a strict hierarchy of problems inside the second level of the $\tilde{\Sigma}_m^{\mathit{plog}}$ hierarchy.

\begin{theorem}\label{Th:sigmatwo}
    For every $k > 1$,  $\atime{\log^k n}{2} \subsetneq \atime{\log^{k+1} n}{2}$.
\end{theorem}

\begin{proof}
Lemma~\ref{lemma:sigma2} proves that $\exactlyonce{k+1} \in \atime{\log^{k+1} n}{2}$. Regarding the lower bound, we will show that $\exactlyonce{k+1}$ (see Problem~\ref{p3}) is \emph{not} in $\atime{\log^k n}{2}$.

Let us assume for the sake of contradiction that there is an alternating random-access Turing machine $M$ that decides $\exactlyonce{k+1}$ in $\atime{\log^k n}{2}$. Let us further assume, w.l.o.g., that every final state of $M$ is universal.
Let $M$ work in time $\lceil \log n \rceil^k \cdot c$ for some constant $c$. Take a binary string $s$ of the form $0^{\lceil \log n \rceil^{k+1}}10^{\lceil \log n \rceil^{k+1}}1^{n- 2 \cdot \lceil \log n \rceil^{k+1} -1}$ such that $\lceil \log n \rceil^{k+1} > \lceil \log n \rceil^k \cdot c$. Let $\mA$ be its corresponding word model.
From our assumption that $M$ decides $\exactlyonce{k+1}$, we get that there is a rejecting computation $\rho$ of $M(\mA)$.  Since every computation of $M$ which rejects $\mA$ must do so reading at most $\lceil \log n \rceil^k \cdot c$ cells, then there must be a position $i$ among the first $\lceil \log n \rceil^{k+1}$ cells in the encoding of $R_0^{\mA}$ in the input tape that was not read during computation $\rho$.
 Define a string $s' = 0^i10^{{\lceil \log n \rceil^{k+1}} - i - 1}10^{\lceil \log n \rceil^{k+1}}1^{n- 2\cdot \lceil \log n \rceil^{k+1} -1}$ and a corresponding word model $\mB$. Clearly, the rejecting computation $\rho$ of $M(\mA)$ is also a rejecting computation of $M(\mB)$. This contradicts the assumption that $M$ decides $\exactlyonce{k+1}$, since $s'$ has exactly one substring $0^{\lceil\log n\rceil^{k+1}}$ and should then be accepted by all computations of $M$.   
\qed
\end{proof}

The following result, together with Theorems~\ref{strictNonD} and~\ref{Th:sigmatwo}, shows that there is a proper hierarchy of problems for every level of the polylogarithmic time hierarchy $\polysigma{m}$. 

\begin{theorem}\label{ThForSigma}
For every $m > 2$ and every $k > 1$, it holds that $\atime{\log^k n}{m} \subsetneq \atime{\log^{k+2} n}{m}$.
\end{theorem}

\begin{proof}
Since $m > 2$, we have that Lemma~\ref{lemma:genconseq} proves that if $m$ is odd, then $\genconseq{k+1}{(m - 1)/2}$ is in $\atime{\log^{k+2} n}{m}$. Likewise, Lemma~\ref{lemma:genexactly} proves that if $m$ is even, then $\genexactly{k+1}{(m - 2)/2}$ is in $\atime{\log^{k+2} n}{m}$. Regarding the lower bounds, it is easy to see (given our previous results in this section) that: (a) for odd $m$, $\genconseq{k+1}{(m - 1)/2}$ is \emph{not} in $\atime{\log^{k} n}{m}$, and (b) for even $m$, $\genexactly{k+1}{(m - 2)/2}$ is also \emph{not} in $\atime{\log^{k} n}{m}$. Note that if $m$ is odd, then we can prove (a) by contradiction following a similar argument than in the proof of the lower bound for Theorem~\ref{strictNonD}. Likewise, if $m$ is even, then we can prove (b) by contradiction following a similar argument than in the proof of Theorem~\ref{Th:sigmatwo}. 
\qed
\end{proof}

It is clear that by taking the complements of the problems $\genconseq{k}{l}$ and $\genexactly{k}{l}$, a similar result holds for each level of the $\piplog{m}$ hierarchy. 

\begin{theorem}\label{ThForPi}
For $m = 2$ and every $k > 1$, it holds that $\atimeop{\log^k n}{m} \subsetneq \atimeop{\log^{k+1} n}{m}$. Moreover, For every $m > 2$ and every $k > 1$, it holds that $\atimeop{\log^k n}{m} \subsetneq \atimeop{\log^{k+2} n}{m}$.
\end{theorem}

\section{On polylogarithmic-time and complete problems}\label{sec:complete}
In this section we investigate whether the concept of complete problem can somehow be applied to the complexity classes \dpolylog and \npolylog. That is, we want to know whether we can isolate the most difficult problems inside these sublinear time complexity classes. The first step towards this objective is to find a suitable concept of many-one reducibility (m-reducibility for short). 

It is quite clear that m-reductions with sublinear time bounds do not work. Consider for instance \dpolylog reductions. 
Assume there is a complete problem $P$ for the class \npolylog under \dpolylog reductions. Let $P'$ belong to \npolylog and let $M$ be a deterministic random-access Turing machine that reduces $P'$ to $P$ in time $c' \cdot log^{k'} n$ for some constant $c'$. Then the output of $M$ given an instance of $P'$ of length $n$ has maximum length $c' \cdot \log^{k'} n$. This means that,  given an input of length $n$ for $P'$ and its reduction, the random-access Turing machine that computes the complete problem $P$ can actually compute $P(s)$ in time $O((\log \log n)^k)$ for some fixed $k$. This is already highly unlikely. If as one would expect there are more than a single complete problem for the class, then we could keep applying reductions from one problem to the other, infinitely reducing the time required to compute the original problem.     

Let us then consider the standard concept of Karp reducibility, i.e., deterministic polynomially bounded many-one reducibility, so that we can avoid the obvious problem described in the previous paragraph. Rather surprisingly, there is no complete problems for \dpolylog and \npolylog, even under these rather expensive reductions for the complexity classes at hand. 

\begin{theorem}
\dpolylog does \emph{not} have complete problems under deterministic polynomially bounded many-one reductions. 
\end{theorem}
\begin{proof}
We prove it by contradiction. Assume that there is such a complete problem $P$. Since $P$ is in \dpolylog, then there is a random-access Turing machine $M$ which computes $P$ in time $O(\log^k n)$ for some fixed $k$. Thus $P$ belongs to $\dtime{\log^k n}$. Let us take the problem $\initial{k+1}$ of deciding the language of binary strings which have a prefix of at least $\lceil\log n\rceil^{k+1}$ consecutive zeros. Since $P$ is complete for the whole class \dpolylog, there must be a function $f: \{0,1\}^* \rightarrow \{0,1\}^*$, computable in polynomial-time, such that $x \in \initial{k+1}$ iff $f(x) \in P$ holds for all $x \in \{0,1\}^*$. It then follows that the size of $f(x)$ is polynomial in the size of $x$. Let $|f(x)| = |x|^{k'}$, we get that the machine $M$ which computes the complete problem $P$ can also decide $\initial{k+1}$ in time $O(\log^{k} n^{k'}) = O((k' \cdot \log n)^k) = O(\log^k n)$. This contradicts the fact that $\initial{k+1} \not\in \dtime{\log^k n}$ as shown in the proof of Theorem~\ref{strictDet}.
\qed
\end{proof}

Using a similar proof than in the previous theorem for \dpolylog, we can prove that the same holds for \npolylog. In fact, we only need to replace the problem $\initial{k+1}$ by $\conseq{k+1}$ and the reference to Theorem~\ref{strictDet} by a reference to Theorem~\ref{strictNonD} in the previous proof, adapting the argument accordingly. 

\begin{theorem}
\npolylog does \emph{not} have complete problems under deterministic polynomially bounded many-one reductions. 
\end{theorem}

Moreover, using the problems $\genconseq{k}{l}$ and $\genexactly{k}{l}$ together with its complements and Theorems~\ref{ThForSigma} and~\ref{ThForPi}, it is easy to prove that the same holds for every individual level of the polylogarithmic time hierarchy.

\begin{theorem}
For every $m \geq 1$, $\sigmaplog{m}$ and $\piplog{m}$ do \emph{not} have complete problems under deterministic polynomially bounded many-one reductions. 
\end{theorem}

\section{The polylogarithmic time hierarchy is strict}\label{polyloghierarchyisstrict}

The problem of whether the polynomial time hierarchy is proper, is one of the oldest and more difficult open problems in complexity theory. In this section, we show that by contrast, it is not difficult to see that the analogous polylogarithmic time hierarchy is indeed proper.

We prove next that the second level of the $\sigmaplog{m}$ hierarchy is strictly includes the first level. This showcases how we can use the results in Section~\ref{hierarchies} to prove that the whole hierarchy does not collapses. 

\begin{theorem}\label{strict1}
$\sigmaplog{1}$ is strictly included in $\sigmaplog{2}$.
\end{theorem}

Theorem~\ref{strict1} is a direct consequence of the fact that for every $k$, $\conseq{k+1}$  is \emph{not} in $\ntime{\log^{k}n}$ (see proof of Theorem~\ref{strictDet}) and the following lemma, which shows that in the second level of the polylogarithmic time hierarchy we can decide $\conseq{k}$ (see Problem~\ref{p1}) in quadratic logarithmic time, independently of the value of $k$. 

\begin{lemma}\label{strictNonDalternations}
For every $k$, $\conseq{k}$ (see Problem~\ref{p1}) can be decided in $\atime{\log^{2} n}{2}$.
\end{lemma}

\begin{proof}
We show that there is an alternating random-access Turing machine $M$ which decides $\conseq{k}$ in time $O(\log^2n)$ which alternates only once from existential to universal states. 
Assume the input tape of $M$ encodes a word model $\mA$ of signature $\pi$, i.e., a binary string.
$M$ works by first guessing (in existential states) two position $p_1$ and $p_2$ within the positions encoding $R_0^{\mA}$ such that $p_1 < p_2$, $p_2 \leq n$ and $p_2 - p_1 = (\lceil\log n\rceil)^k$, and then switching to a universal state and checking if for every position $p_j$ which is between $p_1$ and $p_2$, the value of the input tape in position $p_j$ is $1$. Regarding the time complexity of this procedure, the key point to observe here is that clearly the required arithmetic operations can be computed by $M$ in time $O(\log^2 n)$. In particular $M$ can compute the binary representation of $\lceil\log n\rceil^k$ in time $O(\log^2 n)$. This is the case since $M$ can compute $\lceil \log n \rceil$ in binary in time $O(\log n)$ and then can, in time $O(\log^2 n)$, multiply $\lceil \log n \rceil$ by itself $k$-times.   
\qed
\end{proof}

The previous argument can be generalized to the whole polylogarithmic time hierarchy, obtaining that. 

\begin{theorem}\label{strict2}
For every $m \geq 1$, $\sigmaplog{m}$ is strictly included in $\sigmaplog{m+1}$.
\end{theorem}

Theorem~\ref{strict2} is an easy consequence of Lemma~\ref{ultimolemma} (see below) and the following facts:
\begin{itemize}
\item For $l \geq 0$, $\genconseq{k + 1}{l}$ is \emph{not} in $\atime{\log^{k} n}{2 \cdot l + 1}$ (see proofs of Theorems~\ref{strictNonD} and~\ref{ThForSigma}).
\item For $l \geq 0$, $\genexactly{k + 1}{l}$ is \emph{not} in $\atime{\log^{k} n}{2 \cdot l + 2}$ (see proofs of Theorem~\ref{Th:sigmatwo} and~\ref{ThForSigma}).
\end{itemize}

\begin{lemma}\label{ultimolemma}
For every $k \geq 1$ and $l \geq 0$, we get that:
\begin{enumerate}[a.]
\item $\genconseq{k}{l}$ can be decided in $\atime{\log^2 n}{2 \cdot l + 1}$.
\item $\genexactly{k}{l}$ can be decided in $\atime{\log^{2} n}{2 \cdot l + 2}$.
\end{enumerate}
\end{lemma}

The proof of Lemma~\ref{ultimolemma} follows the general idea of the proof of Lemma~\ref{strictNonDalternations}. Due to space limitations, we omit this easy, but technically cumbersome proof.

\section{Concluding Remarks} \label{sec:schluss}

In this paper we showed that none of the classes $\polysigma{m}$ and $\polypi{m}$ ($m \in \mathbb{N}$) in the polylogarithmic time hierarchy has a complete problem. This result follows from the existence of proper hierarchies inside each of the classes. The proof that such hierarchies exist is constructive by defining concrete problems parameterized by $k \in \mathbb{N}$ for each class. For the definition of these concrete problems we exploit the logics capturing $\polysigma{m}$ and $\polypi{m}$, respectively. 

\bibliographystyle{abbrv}
\bibliography{SOpolylog}

\begin{thebibliography}{10}

\bibitem{babai:stoc2016}
L.~Babai.
\newblock Graph isomorphism in quasipolynomial time.
\newblock In {\em Proceedings of the forty-eighth annual ACM symposium on
  Theory of Computing (STOC 2016)}, pages 684--697, 2016.

\bibitem{EbbinghausF95}
H.-D. Ebbinghaus and J.~Flum.
\newblock {\em Finite model theory}.
\newblock Perspectives in Mathematical Logic. Springer, 1995.

\bibitem{FerrarottiGST18}
F.~Ferrarotti, S.~Gonz{\'{a}}lez, K.-D. Schewe, and J.~M. {Turull Torres}.
\newblock The polylog-time hierarchy captured by restricted second-order logic.
\newblock In {\em 20th International Symposium on Symbolic and Numeric
  Algorithms for Scientific Computing, {SYNASC} 2018, Timisoara, Romania,
  September 20-23, 2018}, pages 133--140. {IEEE}, 2018.

\bibitem{FGST18}
F.~Ferrarotti, S.~Gonz{\'{a}}lez, K.-D. Schewe, and J.~M. {Turull Torres}.
\newblock The polylog-time hierarchy captured by restricted second-order logic.
\newblock {\em CoRR}, abs/1806.07127, 2018.

\bibitem{FerrarottiGST19}
F.~Ferrarotti, S.~Gonz{\'{a}}lez, K.-D. Schewe, and J.~M. {Turull Torres}.
\newblock A restricted second-order logic for non-deterministic
  poly-logarithmic time.
\newblock {\em To appear in the Logic Journal of the {IGPL}}, 2019.

\bibitem{FerrarottiGTBV19}
F.~Ferrarotti, S.~Gonz{\'{a}}lez, J.~M. {Turull Torres}, J.~{Van den Bussche},
  and J.~Virtema.
\newblock Descriptive complexity of deterministic polylogarithmic time.
\newblock In {\em Logic, Language, Information, and Computation - 26th
  International Workshop, WoLLIC 2019, Utrecht, The Netherlands, July 2-5,
  2019, Proceedings}, volume 11541 of {\em Lecture Notes in Computer Science},
  pages 208--222. Springer, 2019.

\bibitem{FerrarottiGTBV19b}
F.~Ferrarotti, S.~Gonz{\'{a}}lez, J.~M. {Turull Torres}, J.~{Van den Bussche},
  and J.~Virtema.
\newblock Descriptive complexity of deterministic polylogarithmic time and
  space.
\newblock {\em Submitted for publication}, 2019.

\bibitem{Immerman99}
N.~Immerman.
\newblock {\em Descriptive complexity}.
\newblock Graduate texts in computer science. Springer, 1999.

\bibitem{barrington:sct1992}
D.~A. {Mix Barrington}.
\newblock Quasipolynomial size circuit classes.
\newblock In {\em Proceedings of the Seventh Annual Structure in Complexity
  Theory Conference, Boston, Massachusetts, USA, June 22-25, 1992}, pages
  86--93. {IEEE} Computer Society, 1992.

\bibitem{barrington:jcss1990}
D.~A. {Mix Barrington}, N.~Immerman, and H.~Straubing.
\newblock On uniformity within {NC}$^1$.
\newblock {\em J. Comput. Syst. Sci.}, 41(3):274--306, 1990.

\end{thebibliography}

\end{document}